\documentclass[a4paper,english,titlepage,12pt]{article}

\usepackage{babel}\usepackage{amsmath,amssymb,amsthm}
\usepackage[numbers,compress]{natbib}
\usepackage{graphicx}
\usepackage{ctable}
\usepackage{float}
\usepackage{titlesec}
\usepackage{caption}

\makeatletter
\def\@xfootnote[#1]{%
  \protected@xdef\@thefnmark{#1}%
  \@footnotemark\@footnotetext}
\makeatother

\newtheorem{lemma}{Lemma}
\newcommand{\new}{\text{\textit{\textbf{new}}}}
\newcommand{\rop}{\mathcal{R}}
\newcommand{\ffdeg}{\mathrm{Deg}}
\newcommand{\forests}{\mathfrak{F}}
\newcommand{\infragr}{\mathfrak{I}}
\newcommand{\forestsmax}{\mathfrak{F}_{\text{max}}}
\newcommand{\iclos}{\mathrm{IClos}}
\newcommand{\electrons}{e}
\newcommand{\sechains}{\mathrm{SE}}
\newcommand{\csat}{C_{\text{sat}}}
\newcommand{\cbig}{C_{\text{big}}}
\newcommand{\cadd}{C_{\text{add}}}

\renewcommand{\thesection}{\Roman{section}}

\makeatletter
\renewcommand{\p@subsection}{\thesection.}
\makeatother
\titlelabel{\thetitle.\ }

\begin{document}

\begin{center}\LARGE{\textbf{New method of computing the contributions of graphs without
lepton loops to the electron anomalous magnetic moment in QED}}
\end{center}
\begin{center}
\large{Sergey Volkov\footnote{E-mail: \texttt{volkoff\underline{
}sergey@mail.ru}}}
\\ \emph{\small{SINP MSU, Moscow, Russia}}
\end{center}

\small{This paper presents a new method of numerical computation of
the mass-independent QED contributions to the electron anomalous
magnetic moment which arise from Feynman graphs without closed
electron loops. The method is based on a forest-like subtraction
formula that removes all ultraviolet and infrared divergences in
each Feynman graph before integration in Feynman-parametric space.
The integration is performed by an importance sampling Monte-Carlo
algorithm with the probability density function that is constructed
for each Feynman graph individually. The method is fully automated
at any order of the perturbation series. The results of applying the
method to 2-loop, 3-loop, 4-loop Feynman graphs, and to some
individual 5-loop graphs are presented, as well as the comparison of
this method with other ones with respect to Monte Carlo convergence
speed.} \normalsize

\section{INTRODUCTION}

The electron anomalous magnetic moment (AMM) is known with a very
high precision. In the experiment~\cite{experiment} the value
$$
a_e=0.00115965218073(28)
$$
was obtained. So, an extremely high accuracy is needed also from
theoretical predictions.

The most precise prediction of electron's AMM at the present time
~\cite{kinoshita_10_new} has the following representation:
$$
a_e=a_e(\text{QED})+a_e(\text{hadronic})+a_e(\text{electroweak}),
$$
$$
a_e(\text{QED})=\sum_{n\geq 1} \left(\frac{\alpha}{\pi}\right)^n
a_e^{2n},
$$
$$
a_e^{2n}=A_1^{(2n)}+A_2^{(2n)}(m_e/m_{\mu})+A_2^{(2n)}(m_e/m_{\tau})+A_3^{(2n)}(m_e/m_{\mu},m_e/m_{\tau}),
$$
where $m_e$, $m_{\mu}$, $m_{\tau}$ are masses of electron, muon, and
tau lepton, respectively. The corresponding numerical value \
\begin{equation}\label{eq_kinoshita_amm}
a_e=0.001159652181643(25)(23)(16)(763)
\end{equation}
was obtained by using the fine structure constant
$\alpha^{-1}=137.035999049(90)$ that had been measured in the recent
experiments with rubidium atoms (see ~\cite{rubidium,codata}). Here,
the first, second, third, and fourth uncertainties come from
$A_1^{(8)}$, $A_1^{(10)}$,
$a_e(\text{hadronic})+a_e(\text{electroweak})$ and the
fine-structure constant\footnote{So, the calculated coefficients are
used for improving the accuracy of $\alpha$.} respectively. Thus, a
still relevant problem is to compute $A_1^{(2n)}$ with a maximum
possible accuracy. The values
$$
A^{(2)}_1=0.5,
$$
$$
A_1^{(4)}=-0.328478965579193\ldots,
$$
$$
A_1^{(6)}=1.181241456\ldots
$$
are known from the analytical results in
~\cite{schwinger1,schwinger2}, ~\cite{analyt2_p,analyt2_z},
~\cite{analyt3}, respectively\footnote{The value for $A_1^{(6)}$ was
a product of efforts of many scientists. See, for example,
~\cite{analyt_mi,analyt_b2,analyt_b3,analyt_b1,analyt_b4,
analyt_b,analyt_e,analyt_d,analyt_c,analyt_ll,analyt_f}.}. The
values
$$
A_1^{(8)}=-1.91298(84),\quad A_1^{(10)}=7.795(336).
$$
were presented by T.Kinoshita et al. in ~\cite{kinoshita_10_new}.
The first one was recently confirmed and improved by S.Laporta using
semi-analytical computation ~\cite{laporta_8}:
$$
A_1^{(8)}=-1.9122457649\ldots.
$$
Thus, the precision of (\ref{eq_kinoshita_amm}) can be slightly
improved. At the present time, there are no independent calculations
of $A_1^{(10)}$.

This paper presents a method of computing the contribution of
Feynman graphs without lepton loops to $A_1^{(2n)}$. We denote this
contribution by $A_1^{(2n)}[\text{no lepton loops}]$. The method
consists of two parts: the subtraction procedure for removal of UV
and IR divergences in Feynman-parametric space before integration
and the graph-specific importance sampling Monte Carlo integration.

The subtraction procedure was presented in ~\cite{volkov_2015}. It
is briefly described in Section \ref{sec_integrands}. This procedure
eliminates IR and UV divergences in each AMM Feynman graph
point-by-point, before integration, in the spirit of papers
~\cite{kinoshita_10_new, levinewright, carrollyao, carroll,
kinoshita_6, kompaniets, bogolubovparasuk, hepp, zavialovstepanov,
scherbina, zavialov, smirnov} etc. This property is substantial for
many-loop calculations when reducing an amount of the needed
computer resources is of critical importance. Let us remark that
$A_1^{(2n)}$ is free from infrared divergences since they are
removed by the on-shell renormalization as well as the ultraviolet
ones (see a more detailed explanation in ~\cite{volkov_2015}).
However, the standard subtractive on-shell renormalization can't
remove IR divergences in Feynman-parametric space before integration
as well as it does for UV divergences\footnote{Moreover, it can
generate additional IR-divergences, see a more detailed explanation
in ~\cite{volkov_2015}.}. The structure of IR and UV divergences in
individual Feynman graphs is quite complicated. IR and UV
divergences can be, in a certain sense, entangled\footnote{If $G'$
is a vertex-like (see section \ref{subsec_preliminary}) subgraph of
a graph $G$, this subgraph contains the vertex that is incident to
the external photon line of $G$, and the electron path connecting
the external electron lines of $G$ passes through $G'$, then the
Feynman amplitude of $G'$ is \textquotedblleft
enhanced\textquotedblright\ by an IR divergent multiplier, see
~\cite{weinberg,yennie}. However, if the Feynman amplitude of $G'$
had already been UV-divergent, then we can observe an
\textquotedblleft entanglement\textquotedblright\ of UV and IR
divergences. For example, see the expressions for 2-loop
renormalization constants from ~\cite{entangle} that were obtained
using dimensional regularization
 to control UV divergences and a photon mass $\lambda$ to control IR divergences:
 that expressions contain terms like $\ln (\lambda/m)/\epsilon$ together with
 terms like $1/\epsilon^2$ and $\ln^2(\lambda/m)$, where $\epsilon$ is the
 parameter of dimensional regularization.
 These terms remain after summing all 2-loop Feynman graphs.}
 with each other. Therefore, a special procedure is required for
removing both UV and IR divergences. Let us recapitulate the
advantages of the developed subtraction procedure.
\begin{enumerate}
\item It is fully automated for any $n$.
\item It is comparatively easy for realization on computers.
\item It can be represented as a forest-like formula. This formula differs from
the forest formula of Zavialov and Stepanov~\cite{zavialovstepanov},
Scherbina~\cite{scherbina}, and Zimmermann~\cite{zimmerman} only in
the choice of linear operators and in the way of combining them.
\item The contribution of each Feynman graph to $A_1^{(2n)}$ can be represented as a single
Feynman-parametric integral. The value of $A_1^{(2n)}$ is the sum of
these contributions.
\item Feynman parameters can be used directly, without any
additional tricks.
\end{enumerate}
See a detailed description in ~\cite{volkov_2015}. The subtraction
procedure was checked independently by F. Rappl using Monte Carlo
integration based on Markov chains ~\cite{rappl}.

All Feynman-parametric integrals are finite after applying the
subtraction procedure. However, the integrands remain badly-behaved:
they have a steep landscape, peaks and integrable singularities.
This fact makes it difficult to integrate with a high accuracy when
the number of dimensions is large (for example, for $n=5$ we have
$13$ dimensions or even more). The known universal integration
routines can solve this problem only partially and often not
satisfactorily. However, the simplicity of the subtraction procedure
makes it possible to understand the behaviour of integrands and to
develop an importance sampling Monte Carlo algorithm based on this
known behaviour. The algorithm for the integrands corresponding to
AMM Feynman graphs without lepton loops is presented in Section
\ref{sec_carlo}. This algorithm is based on the ideas that were used
by different scientists for proving UV-finiteness of renormalized
Feynman amplitudes ~\cite{speer,azp}. The probability density
function is constructed for each Feynman graph
individually\footnote{but fully automatically}. For constructing the
probability density function we use the ultraviolet degrees of
divergence of the so-called \emph{I-closures} of sets of graph
internal lines. The notion of I-closure is first introduced in this
paper. It was observed that the behaviour of on-shell
Feynman-parametric integrands is well approximated using I-closures.
The developed importance sampling integration can be combined with
splitting-based adaptive algorithms. A variant of an adaptive
algorithm is provided (see Section \ref{subsec_technical}). The new
integration algorithm has the following advantages:
\begin{itemize}
\item fast convergence
\item reliable error estimation on early stages of calculation
\item relatively small part of samples requires increased arithmetic
precision for preventing round-off errors (compared to the method
from ~\cite{volkov_2015}, for example)
\end{itemize}
The techniques for stabilizing and for preventing error
underestimation are presented (Section \ref{subsec_stabil}). This
Monte Carlo algorithm can also be used for integrating other
Feynman-parametric integrals provided that we have the needed
information about the integrand behaviour.

Numerical calculation results that were obtained on a personal
computer are presented in Section \ref{sec_results}:
$A_1^{(2n)}[\text{no lepton loops}]$ for $n=2,3,4$, the
contrubutions of the ladder graphs and the fully crossed ladder
graphs up to 5 loops. Each value is given with the error estimation
and with the number of Monte Carlo samples. The comparison of this
results with known ones with respect to values (Section
\ref{subsec_results}) and Monte Carlo convergence speed (Section
\ref{subsec_speed}) is presented. The results for the 4-loop and
5-loop fully crossed ladder graphs are new.

\section{CONSTRUCTION OF THE INTEGRANDS}

\subsection{Preliminary remarks}\label{subsec_preliminary}

We will work in the system of units, in which $\hbar=c=1$, the
factors of $4\pi$ appear in the fine-structure constant:
$\alpha=e^2/(4\pi)$, the tensor $g_{\mu\nu}$ is defined by
$$
g_{\mu\nu}=g^{\mu\nu}=\left(\begin{matrix}1 & 0 & 0 & 0 \\ 0 & -1 &
0 & 0 \\ 0 & 0 & -1 & 0 \\ 0 & 0 & 0 & -1 \end{matrix}\right),
$$
the Dirac gamma-matrices satisfy the following condition
$\gamma^{\mu}\gamma^{\nu}+\gamma^{\nu}\gamma^{\mu}=2g^{\mu\nu}$.

We will use Feynman graphs with the propagators
\begin{equation}\label{eq_electron_propagator}
\frac{i(\hat{p}+m)}{p^2-m^2+i\varepsilon}
\end{equation} for
electron lines and
\begin{equation}\label{eq_feynman_gauge}
\frac{-g_{\mu\nu}}{p^2+i\varepsilon}
\end{equation}
for photon lines. It is always presupposed that a Feynman graph is
strongly connected and doesn't have electron loops with odd number
of lines.

The number $\omega(G)=4-N_{\gamma}-\frac{3}{2}N_e$ is called the
\emph{ultraviolet degree of divergence} of the graph $G$. Here,
$N_{\gamma}$ is the number of external photon lines of $G$, $N_e$ is
the number of external electron lines of $G$.

If for some subgraph\footnote{In this paper we take into account
only such subgraphs that are strongly connected and contain all
lines that join the vertexes of the given subgraph.} $G'$ of the
graph $G$ the condition $\omega(G')\geq 0$ is satisfied, then
UV-divergence can appear. A graph $G'$ is called UV-divergent if
$\omega(G')\geq 0$. There are the following types of UV-divergent
subgraphs in QED Feynman graphs: \emph{electron self-energy
subgraphs} ($N_e=2,N_{\gamma}=0$), \emph{vertex-like} subgraphs
($N_e=2,N_{\gamma}=1$), \emph{photon self-energy subgraphs}
($N_e=0,N_{\gamma}=2$), \emph{photon-photon scattering
subgraphs}\footnote{The divergences of this type vanish in the sum
of all Feynman graphs, but they can arise in individual graphs.}
($N_e=0,N_{\gamma}=4$).

\subsection{The subtraction procedure for calculating $A_1^{(2n)}$}

The definitions in this section repeat the ones given in
~\cite{volkov_2015}.

Two subgraphs are said to overlap if they are not contained one
inside the other, and their sets of lines have a non-empty
intersection.

A set of subgraphs of a graph is called a \emph{forest} if any two
elements of this set don't overlap.

For a vertex-like graph $G$ by $\forests[G]$ we denote the set of
all forests $F$ consisting of UV-divergent subgraphs of $G$ and
satisfying the condition $G\in F$. By $\infragr[G]$ we denote the
set of all vertex-like subgraphs $G'$ of $G$ such that $G'$ contains
the vertex that is incident\footnote{We say that a line $l$ and a
vertex $v$ are \emph{incident} if $v$ is one of the endpoints of
$l$.} to the external photon line of $G$.\footnote{In particular,
$G\in \infragr[G]$.}

Let us define the following linear operators that are applied to the
Feynman amplitudes of UV-divergent subgraphs:
\begin{enumerate}
\item $A$ --- the projector of AMM. This
operator is applied to the Feynman amplitudes of vertex-like
subgraphs. Let $\Gamma_{\mu}(p,q)$ be the Feynman amplitude
respective to an electron of initial and final four-momenta $p-q/2$,
$p+q/2$. The Feynman amplitude $\Gamma_{\mu}$ can be expressed in
terms of three form-factors:
$$
\overline{u}_2 \Gamma_{\mu}(p,q) u_1 = \overline{u}_2 \left(
f(q^2)\gamma_{\mu} -\frac{1}{2m}g(q^2)\sigma_{\mu\nu}q^{\nu} +
h(q^2)q_{\mu} \right) u_1,
$$
where $(p-q/2)^2=(p+q/2)^2=m^2$,
$(\hat{p}-\hat{q}/2-m)u_1=\overline{u}_2(\hat{p}+\hat{q}/2-m)=0$,
$$
\sigma_{\mu\nu}=\frac{1}{2}(\gamma_{\mu}\gamma_{\nu}-\gamma_{\nu}\gamma_{\mu}),
$$
see, for example, \cite{ll4}. By definition, put
\begin{equation}\label{eq_a_def}
A \Gamma_{\mu} = \gamma_{\mu}\cdot \lim_{q^2\rightarrow 0} g(q^2).
\end{equation}
\item The definition of the operator $U$ depends on the type of
UV-divergent subgraph to which the operator is applied:
\begin{itemize}
\item If $\Pi$ is the Feynman amplitude corresponding to a photon
self-energy subgraph or a photon-photon scattering subgraph, then,
by definition, $U\Pi$ is the Taylor expansion of $\Pi$ around zero
momenta up to the UV divergence degree of this subgraph.
\item If $\Sigma(p)$ is the Feynman amplitude that corresponds to an
electron self-energy subgraph,
\begin{equation}\label{eq_sigma_general}
\Sigma(p)=a(p^2)+b(p^2)\hat{p},
\end{equation}
then, by definition\footnote{Note that it differs from the standard
on-shell renormalization.},
$$
U\Sigma(p) = a(m^2)+b(m^2)\hat{p}.
$$
\item If $\Gamma_{\mu}(p,q)$ is the Feynman amplitude corresponding to a
vertex-like subgraph,
\begin{equation}\label{eq_gamma_general_q0}
\Gamma_{\mu}(p,0)=a(p^2)\gamma_{\mu} + b(p^2)p_{\mu} +
c(p^2)\hat{p}p_{\mu}+d(p^2)(\hat{p}\gamma_{\mu}-\gamma_{\mu}\hat{p}),
\end{equation}
then, by definition,
\begin{equation}\label{eq_u_vertex}
U\Gamma_{\mu}=a(m^2) \gamma_{\mu}.
\end{equation}
\end{itemize}
\item $L$ is the operator that is used in the standard subtractive on-shell renormalization
of vertex-like subgraphs. If $\Gamma_{\mu}(p,q)$ is the Feynman
amplitude that corresponds to a vertex-like subgraph,
$$
\Gamma_{\mu}(p,0)=a(p^2)\gamma_{\mu} + b(p^2)p_{\mu} +
c(p^2)\hat{p}p_{\mu}+d(p^2)(\hat{p}\gamma_{\mu}-\gamma_{\mu}\hat{p}),
$$
then, by definition,
\begin{equation}\label{eq_q_def}
L\Gamma_{\mu}=[a(m^2)+mb(m^2)+m^2c(m^2)]\gamma_{\mu}.
\end{equation}
\end{enumerate}

Let $f_G$ be the unrenormalized Feynman amplitude that corresponds
to a vertex-like graph $G$. By definition, put
\begin{equation}\label{eq_rop_tilde}
\tilde{f}_G=\rop^{\new}_G f_G,
\end{equation}
where
\begin{equation}\label{eq_rop}
\rop^{\new}_G=\sum_{\substack{F=\{G_1,\ldots,G_n\}\in \forests[G] \\
G'\in \infragr[G]\cap F}}(-1)^{n-1}M^{G'}_{G_1}M^{G'}_{G_2}\ldots
M^{G'}_{G_n},
\end{equation}
\begin{equation}\label{eq_operators}
M^{G'}_{G''}=\begin{cases}A_{G'},\text{ if }G'=G'', \\
U_{G''},\text{ if }G''\notin \infragr[G]\text{, or }G''\varsubsetneq
G',
\\ L_{G''},\text{ if }G''\in \infragr[G], G'\varsubsetneq G'', G''\neq
G,
\\ (L_{G''}-U_{G''}),\text{ if }G''=G, G'\neq G.\end{cases}
\end{equation}
In this notation, the subscript of an operator symbol denotes the
subgraph to which this operator is applied.

By $\check{f}_G$ we denote the coefficient before $\gamma_{\mu}$ in
$\tilde{f}_G$. The value $\check{f}_G$ is the contribution of the
graph $G$ to the AMM:
$$
a_{e,1}^{\new}=\sum_{G} \check{f}_G,
$$
where the summation goes over all vertex-like Feynman graphs.

For example, for the graph $G$ from FIG. \ref{fig_example_operators}
we have
$$
\infragr[G]=\{G,bcd\}
$$
(subgraphs are specified by enumeration of vertexes). Also, we have
two other vertex-like UV-divergent subgraphs $efg$, $fgh$, one
electron self-energy subgraph $efgh$. Thus,
$$
\tilde{f}_G =
\left[A_G(1-U_{bcd})-(L_G-U_G)A_{bcd}\right]\left(1-U_{efgh}\right)\left(1-U_{efg}-U_{fgh}\right)f_G.
$$

\begin{figure}[H]
\begin{center}
\includegraphics[scale=0.5]{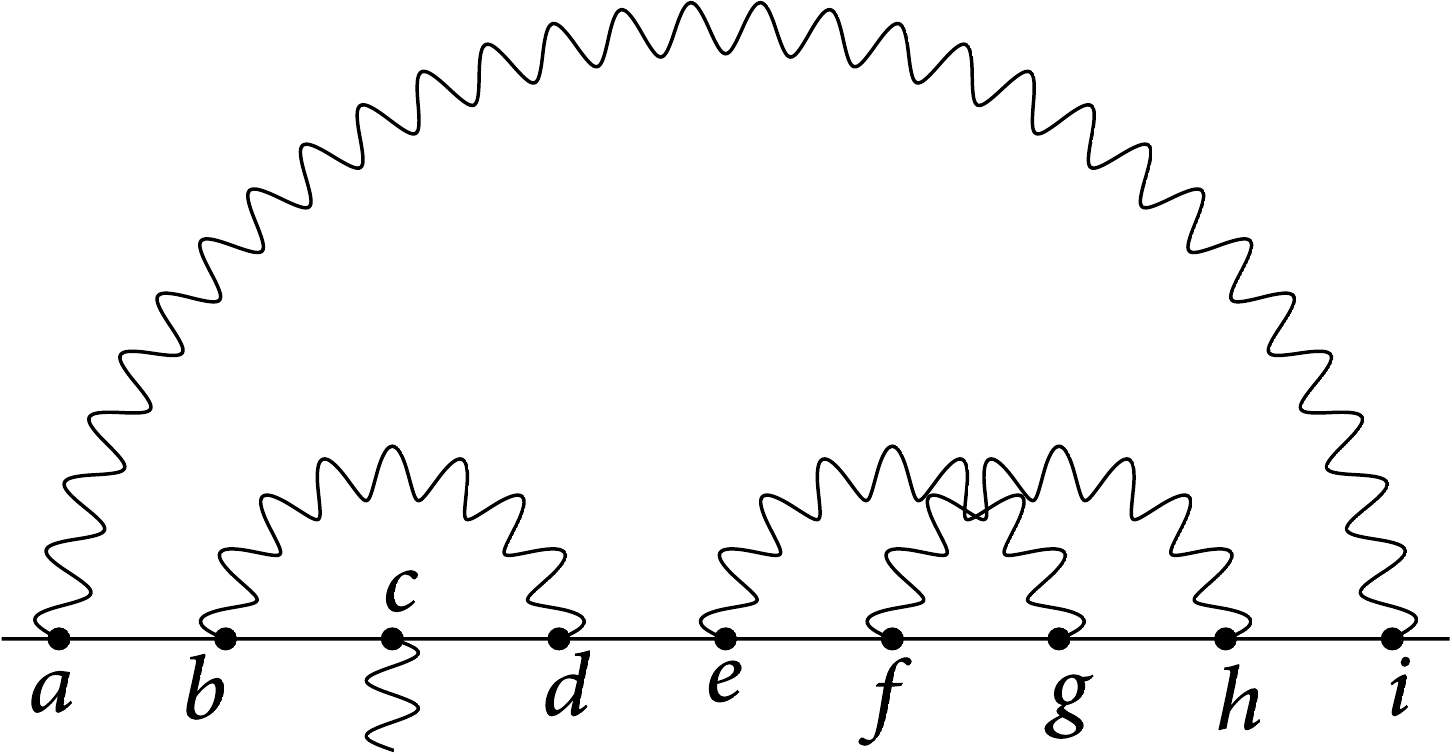}
\caption{Example of an AMM Feynman graph.}
\label{fig_example_operators}
\end{center}
\end{figure}

It is known~\cite{volkov_2015} that
$$
a_{e,1}^{\new}=a_{e,1},
$$
where
$$
a_{e,1}=\sum_{n\geq 1}\left(\frac{\alpha}{\pi}\right)^n A_1^{(2n)}.
$$
If we sum only over graphs with a fixed number of vertices, we can
obtain the corresponding term $A_1^{(2n)}$. Also, summing only over
graphs without electron loops, we obtain $A_1^{(2n)}[\text{no lepton
loops}]$ (the proof of this fact is the same as for $A_1^{(2n)}$,
but only we should restrict the set of graphs that are considered in
the proof to the gauge-invariant set of all graphs without lepton
loops).

\subsection{Integrands in Feynman-parametric
space}\label{sec_integrands}

Calculation of a graph $G$ contribution to $a_{e,1}^{\new}$ can be
reduced to the Feynman-parametric integration
\begin{equation}\label{eq_feyn_param_integral}
\int_{z_1,\ldots,z_n>0}I(z_1,\ldots,z_n)\delta(z_1+\ldots+z_n-1)
dz_1\ldots dz_n.
\end{equation}
To obtain the integrand value $I(z_1,\ldots,z_n)$ for given values
of Feynman parameters $z_1,\ldots,z_n$ we should perform the
following steps.
\begin{enumerate}
\item Using the propagators
$$
(\hat{p}+m)e^{iz_j(p^2-m^2+i\varepsilon)},\quad ig_{\mu\nu}
e^{iz_j(p^2+i\varepsilon)}
$$
instead of (\ref{eq_electron_propagator}), (\ref{eq_feynman_gauge}),
applying the subtraction procedure and performing momentum
integrations we obtain $\check{f}_G(\underline{z},\varepsilon)$,
where $\underline{z}=(z_1,z_2,\ldots,z_n)$. The momentum integration
is carried out by using explicit formulas for integrals of
multi-dimensional Gaussian functions multiplied by polynomials.
\item Put
$$
I(z_1,\ldots,z_n)=\lim_{\varepsilon\rightarrow +0} \int_0^{+\infty}
\lambda^{n-1} \check{f}_G(z_1\lambda,\ldots,z_n\lambda,\varepsilon)
d\lambda.
$$
The integration with respect to $\lambda$ is performed analytically
by using the formula
$$
\int_0^{+\infty} \lambda^{D-1}e^{\lambda(ik-\varepsilon)}d\lambda =
\frac{(D-1)!}{(\varepsilon-ik)^D}.
$$
\end{enumerate}
The integral (\ref{eq_feyn_param_integral}) is suitable for
numerical integration. A detailed description can be found in
~\cite{volkov_2015}.

\section{MONTE CARLO INTEGRATION}\label{sec_carlo}

\subsection{Importance Sampling}\label{subsec_mc_importance}

For integration of a function $f(\underline{x})=f(x_1,\ldots,x_n)$
over $\Omega$ using Monte-Carlo approach with the probability
density function $g(\underline{x})$, $\int_{\Omega}
g(\underline{x})d\underline{x}=1$, we take randomly $N$ samples
$\underline{x}_1,\ldots,\underline{x}_N$ with the distribution $g$
and approximate the needed integral by
$$
\frac{1}{N} \sum_{j=1}^N
\frac{f(\underline{x}_j)}{g(\underline{x}_j)}.
$$
The standard deviation of this value is
\begin{equation}\label{eq_sigma}
\sigma=\sqrt{\frac{V(f,g)}{N}},
\end{equation}
where
$$
V(f,g)=\int_{\Omega} \frac{f(\underline{x})^2}{g(\underline{x})}
d\underline{x} - \left(\int_{\Omega} f(\underline{x}) d\underline{x}
\right)^2,
$$
see ~\cite{mc_james}.

When the number of dimensions is large, it is very important to
choose an appropriate function $g(\underline{x})$ for obtaining
accurate results. It is desirable to have this before applying
splitting-based adaptive Monte Carlo routines. For a given function
$g(\underline{x})$ we may have one of the following three
situations.
\begin{enumerate}
\item\label{mc_case_bounded} The function $f(\underline{x})/g(\underline{x})$ is bounded. In
this case, we will have a stable Monte Carlo convergence with the
error that can be approximated by (\ref{eq_sigma}). However, the
convergence may be slow due to the big value of $V(f,g)$.
\item\label{mc_case_finite_unbounded}The function $f(\underline{x})/g(\underline{x})$ is unbounded,
but $V(f,g)$ is finite. In this case, the error can be approximated
by (\ref{eq_sigma}) too. However, the convergence can be unstable.
We should use some techniques for stabilization and adequate error
estimation.
\item\label{mc_case_infinite} $V(f,g)$ is infinite. In this case, we will have unstable
convergence that is slower than $C/\sqrt{N}$. An adequate error
estimation is difficult in this case.
\end{enumerate}
For the Feynman-parametric integrals that are considered in this
paper, the optimal realistic\footnote{It can be proved that the
optimal selection is
$g(\underline{x})=|f(\underline{x})|/(\int_{\Omega}|f(\underline{x})|d\underline{x}$).
However, it is difficult to do a stable generation of random samples
with this distribution.} selection is usually somewhere in case
\ref{mc_case_finite_unbounded}.

Selection of the function $g(\underline{x})$ needs a lot of care.
For example, let $\Omega=[0;1]^n$,
$$
f(x_1,\ldots,x_n)=a_1\ldots a_n x_1^{a_1-1}\ldots x_n^{a_n-1},
$$
\begin{equation}\label{eq_density_bb}
g(x_1,\ldots,x_n)=b_1\ldots b_n x_1^{b_1-1}\ldots x_n^{b_n-1},
\end{equation}
$a_1,\ldots,a_n,b_1,\ldots,b_n>0$. In this case,
\begin{equation}\label{eq_ab_disp}
V(f,g) = \frac{a_1^2\ldots a_n^2}{b_1\ldots b_n (2a_1-b_1)\ldots
(2a_n-b_n)}-1.
\end{equation}
On the one hand, if there exists $j$ such that $b_j>2a_j$, then we
fall into case \ref{mc_case_infinite}. On the other hand, if we take
some small value for all $b_j$, then the value (\ref{eq_ab_disp})
can be very big due to the factor $1/(b_1\ldots b_n)$ when $n$ is
large.

\subsection{Graph-specific probability density
functions}\label{subsec_pdf}

Let us consider an AMM Feynman graph $G$ containing electron and
photon lines and not containing electron loops. Suppose that the
contribution of $G$ is the integral (\ref{eq_feyn_param_integral}),
where $I(z_1,\ldots,z_n)$ is the integrand that is obtained by the
construction that is described above. Let us construct the
probability density function $g(z_1,\ldots,z_n)$ for Monte Carlo
integration. In this case, $g$ must satisfy the condition
$$
\int_{z_1,\ldots,z_n>0}g(z_1,\ldots,z_n)\delta(z_1+\ldots+z_n-1)
dz_1\ldots dz_n=1.
$$

We will use E.Speer's idea~\cite{speer} with some modifications. All
the space $\mathbb{R}^n$ is split\footnote{Let us remark that the
components has intersections on their boundaries. However, this is
inessential for integration.} into \emph{sectors}. Each sector
corresponds to a permutation $(j_1,\ldots,j_n)$ of
$\{1,2,\ldots,n\}$ and is defined by
$$
S_{j_1,\ldots,j_n}=\{(z_1,\ldots,z_n)\in\mathbb{R}:\ z_{j_1}\geq
z_{j_2}\geq\ldots\geq z_{j_n}\}.
$$
We define the function $g_0(z_1,\ldots,z_n)$ on $S_{j_1,\ldots,j_n}$
by the following relation
\begin{equation}\label{eq_mc_g0}
g_0(z_1,\ldots,z_n)=\frac{\prod_{l=2}^n
(z_{j_l}/z_{j_{l-1}})^{\ffdeg(\{j_l,j_{l+1},\ldots,j_n\})}}{z_1z_2\ldots
z_n},
\end{equation}
where $\ffdeg(s)>0$ is defined for each set $s$ of internal
lines\footnote{Note that the sets can be not connected.} of $G$
except the empty set and the set of all internal lines of $G$. The
probability density function is defined by
$$
g(z_1,\ldots,z_n) =
\frac{g_0(z_1,\ldots,z_n)}{\int_{z_1,\ldots,z_n>0}
g_0(z_1,\ldots,z_n)\delta(z_1+\ldots+z_n-1) dz_1\ldots dz_n}.
$$

The numbers $\ffdeg(\{j_l,j_{l+1},\ldots,j_n\})$, $l=2\ldots n$,
play the same role in the sector $S_{j_1,\ldots,j_n}$ as
$b_1,\ldots,b_n$ play in (\ref{eq_density_bb}). Thus, adjusting
$\ffdeg(s)$ requires a lot of care. Let us describe the procedure of
determining $\ffdeg$ for the graph $G$.

Let $s$ be a subset of the set of all internal lines of $G$. Put
$$
\omega(s)=2N_L(s)+|\electrons(s)|/2-|s|,
$$
where $|x|$ is the cardinality of a set $x$, $\electrons(s)$ is the
set of all electron lines in $s$, $N_L(s)$ is the number of
independent loops in $s$. If $s$ is the set of all internal lines of
a subgraph of $G$, then $\omega(s)$ coincides with the ultraviolet
degree of divergence of this subgraph that is defined in Section
\ref{subsec_preliminary}.

By $\iclos(s)$ we denote the set $s\cup s'$, where $s'$ is the set
of all internal photon lines $l$ in $G$ such that $s$ contains the
electron path in $G$ connecting the ends of $l$. The set $\iclos(s)$
is called the \emph{I-closure} of the set $s$. For example, if $G$
is the graph from FIG. \ref{fig_cross}, then we have
$$
\iclos(\{3,5,6\})=\{3,5,6\},
$$
$$
\iclos(\{3,4,5,6\})=\{3,4,5,6,9\},
$$
$$
\iclos(\{2,3,4,5,6,7\})=\{2,3,4,5,6,7,8,9\},
$$
$$
\iclos(\{1,2,3,4,5,6\})=\{1,2,3,4,5,6,7,8,9\}.
$$

By definition, put
$$
\omega'(s)=\omega(\iclos(s)).
$$
For example, for the graph $G$ from FIG. \ref{fig_chains}, we have
$$
\omega'(\{2,4,7,9\})=\omega(\{2,4,7,9,11,12,13,14\})=2.
$$

A graph $G''$ belonging to a forest $F\in\forests[G]$ is called a
\emph{child} of a graph $G'\in F$ in $F$ if $G''\varsubsetneq G'$,
and there is no $G'''\in F$ such that $G'''\varsubsetneq G'$,
$G''\varsubsetneq G'''$.

If $F\in\forests[G]$ and $G'\in F$ then by $G'/ F$ we denote the
graph that is obtained from $G'$ by shrinking all childs of $G'$ in
$F$ to points.

We also will use the symbols $\omega$, $\omega'$ for graphs $G'$
that are constructed from $G$ by some operations like described
above and for sets $s$ that are subsets of the set of internal lines
of the whole graph $G$. We will denote it by $\omega_{G'}(s)$ and
$\omega'_{G'}(s)$, respectively. This means that we apply the
operations $\omega$ and $\omega'$ in the graph $G'$ to the set $s'$
that is the intersection of $s$ and the set of all internal lines of
$G'$. For example, for the graph $G$ from FIG. \ref{fig_select} and
the forests $F_1=\{G,cdef,cde\}$, $F_2=\{G,cdef,def\}$, we have
$$
\omega'_{cdef/ F_1}(\{3,5,7\})=\omega'_{cdef/
F_1}(\{5,7\})=\omega_{cdef/ F_1}(\{5,7\})=1/2,
$$
$$
\omega'_{cde/ F_1}(\{3,5,7\})=\omega'_{cde/ F_1}(\{3\})=\omega_{cde/
F_1}(\{3\})=-1/2,
$$
$$
\omega'_{cdef/ F_2}(\{3,5,7\})=\omega'_{cdef/
F_2}(\{3\})=\omega_{cdef/ F_2}(\{3,6\})=1/2,
$$
$$
\omega'_{def/ F_2}(\{3,5,7\})=\omega'_{def/
F_2}(\{5,7\})=\omega_{def/ F_2}(\{5,7\})=-3/2,
$$
$$
\omega'_{G/ F_1}(\{1,2,9\})=\omega_{G/ F_1}(\{1,2,8,9\})=-1/2.
$$

Electron self-energy subgraphs and lines joining them form chains
$l_1 G_1 l_2 G_2 \ldots l_r G_r l_{r+1}$, where $l_j$ are electron
lines of $G$, $G_j$ are electron self-energy subgraphs of $G$.
Maximal (with respect to inclusion) subsets
$\{l_1,l_2,\ldots,l_{r+1}\}$ corresponding to such chains are called
\emph{SE-chains}. The set of all SE-chains of $G$ is denoted by
$\sechains[G]$. For example, for the graph $G$ from FIG.
\ref{fig_select}, \ref{fig_chains}, we have
$$
\sechains[G]=\{\{2,9\}\},
$$
$$
\sechains[G]=\{\{1,3,5\},\{6,8,10\}\}
$$
respectively. Let us remark that SE-chains never intersect, but the
corresponding chains of electron self-energy subgraphs can be nested
one inside the other.

Suppose a graph $G'$ is constructed from $G$ by operations like
described above; by definition, put
$$
\omega^*_{G'}(s)=\omega'_{G'}(s)+\frac{1}{2}\sum_{\substack{s'\in\sechains[G]
\\ s'\subseteq s,\ s'\text{ in }G'}}(|s'|-1)
$$
(it is important that here we consider the SE-chains of the whole
graph $G$). For example, for the graph $G$ from FIG.
\ref{fig_select}, for $F=\{G,cdef\}$, we have
$$
\omega^*_{G/ F}(\{1,2,6,7,9\})=\omega'_{G/
F}(\{1,2,9\})+1/2=\omega_{G/ F}(\{1,2,8,9\})+1/2=0,
$$
$$
\omega^*_{cdef/ F}(\{1,2,6,7,9\})=\omega'_{cdef/
F}(\{6,7\})=\omega(\{6,7\})=-2,
$$
$$
\omega^*_{G/ F}(\{1,9\})=\omega'_{G/ F}(\{1,9\}) = \omega_{G/
F}(\{1,9\}) = -1,
$$
for the graph $G$ from FIG. \ref{fig_chains} and for
$F=\{G,bc,de,gh,ij\}$, we have
$$
\omega^*_G(\{1,5,8,10\})=\omega'(\{1,5,8,10\})=\omega(\{1,5,8,10\})=-2,
$$
$$
\omega^*_G(\{1,3,5,8,10\})=1+\omega'(\{1,3,5,8,10\})=1+\omega(\{1,3,5,8,10\})=-3/2,
$$
$$
\omega^*_G(\{1,5,6,8,10\})=1+\omega'(\{1,5,6,8,10\})=1+\omega(\{1,5,6,8,10\})=-3/2,
$$
$$
\omega^*_G(\{1,3,5,6,8,10\})=2+\omega'(\{1,3,5,6,8,10\})=2+\omega(\{1,3,5,6,8,10\})=-1,
$$
$$
\omega^*_{G/ F}(\{1,3,5,6,8,10\})=2+\omega'_{G/
F}(\{1,3,5,6,8,10\})=2-2=0,
$$
$$
\omega^*_{bc/ F}(\{1,3,5,6,8,10\})=\omega'_{bc/
F}(\emptyset)=\omega(\emptyset)=0.
$$

By $\forestsmax[G]$ we denote the set of all maximal forests
belonging to $\forests[G]$ (with respect to inclusion). For example,
for $G$ from FIG. \ref{fig_example_operators}, \ref{fig_select},
\ref{fig_cross}, \ref{fig_chains}, we have
$$ \forestsmax[G]=\{\{G,bcd,efgh,efg\},\{G,bcd,efgh,fgh\}\}, $$
$$ \forestsmax[G]=\{\{G,cdef,cde\},\{G,cdef,def\}\}, $$
$$ \forestsmax[G]=\{\{G\}\}, $$
$$ \forestsmax[G]=\{\{G,bc,de,gh,ij\}\} $$
respectively.

Let $\csat\geq 0$, $\cbig>0$, $\cadd>-\csat$ be constants. By
definition, put
\begin{equation}\label{eq_deg}
\ffdeg(s) = \begin{cases} \cbig,\text{ if $s$ contain all electron
lines of $G$,} \\ \cadd + \max\left[ \csat,
\min_{F\in\forestsmax[G]} \sum_{G'\in F} \max(0, -\omega^*_{G'/
F}(s)) \right], \\ \quad \text{otherwise.}\end{cases}
\end{equation}
For example, for the graph $G$ from FIG. \ref{fig_select}, we have
$\forestsmax[G]=\{F_1,F_2\}$, $F_1=\{G,cdef,cde\}$,
$F_2=\{G,cdef,def\}$,
$$
\ffdeg(\{3,7,5\})=\cadd + \max\left[ \csat,\max(0,-\omega^*_{G/
F_1}(\emptyset))\right.
$$
$$ +\min\left(\max(0,-\omega^*_{cdef/ F_1}(\{7,5\})) +
\max(0,-\omega^*_{cde/ F_1}(\{3\})),\right.$$
$$\left.\left. \max(0,-\omega^*_{cdef/ F_2}(\{3\})) +
\max(0,-\omega^*_{def/ F_2}(\{7,5\}))\right)\right]
$$
$$
=\cadd+\max(\csat, \min(1/2,3/2))=\cadd+\max(\csat,1/2),
$$
$$
\ffdeg(\{2,8,9\})=\cadd+\max(\csat,3/2),\quad
\ffdeg(\{1,2,9\})=\cadd+\csat,
$$
$$
\ffdeg(\{1,2,3,4,5,6,9\})=\cbig,
$$
for the graph $G$ from FIG. \ref{fig_chains}, we have
$$
\ffdeg(\{1,3,5,6,8,10,11,12,13,14\})=\cadd+\max(\csat,0+1+1+1+1)
$$ $$=\cadd+\max(\csat,4),
$$
$$
\ffdeg(\{1,2,3,4,5,6,7,8,9,10\})=\cbig.
$$

There are certain theoretical reasons\footnote{Some of theoretical
considerations will be published in the further papers. For the
simple case when there are no UV divergent subgraphs in $G$, and
$\varepsilon>0$ is fixed in (\ref{eq_electron_propagator}) and
(\ref{eq_feynman_gauge}) and quite far from zero, we can use
$\ffdeg(s)=\lceil-\omega(s)\rceil$. However, when
$\varepsilon\rightarrow 0$, an additional divisor vanishing on some
points of the integration area boundary appears in the integrand.
This fact complicates the problem of approximating the integrand. It
was observed that I-closures can be used in this situation. Also,
the existence of divergent subgraphs intricates the problem even
more. At the present moment, there is no mathematical proof that
(\ref{eq_deg}) does not lead to case \ref{mc_case_infinite} from
Section \ref{subsec_mc_importance}.} for using (\ref{eq_deg}). For a
good Monte Carlo convergence we can use the values
\begin{equation}\label{eq_mc_constants}
\cbig=0.475,\quad \csat=0.3,\quad \cadd=0.615.
\end{equation}
These values were obtained by a series of numerical experiments on
4-loop Feynman graphs.

\begin{figure}[H]
\begin{center}
\includegraphics[scale=0.5]{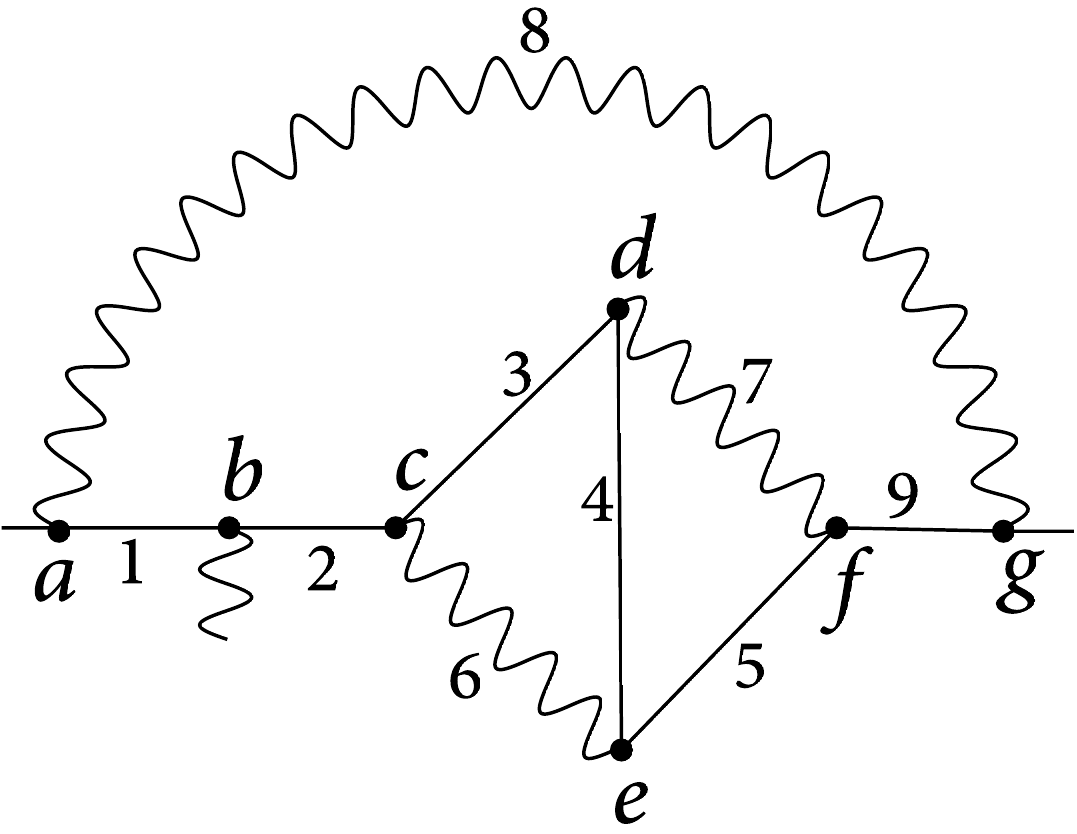}
\caption{Example with overlapping UV-divergent subgraphs.}
\label{fig_select}
\end{center}
\end{figure}

\begin{figure}[H]
\begin{center}
\includegraphics[scale=0.5]{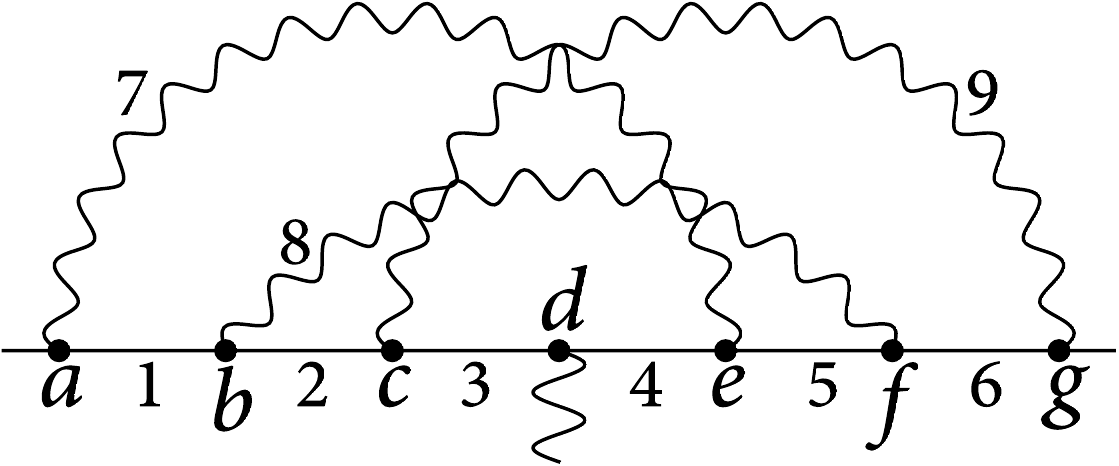}
\caption{Example, 3-loop fully crossed ladder.} \label{fig_cross}
\end{center}
\end{figure}

\begin{figure}[H]
\begin{center}
\includegraphics[scale=0.5]{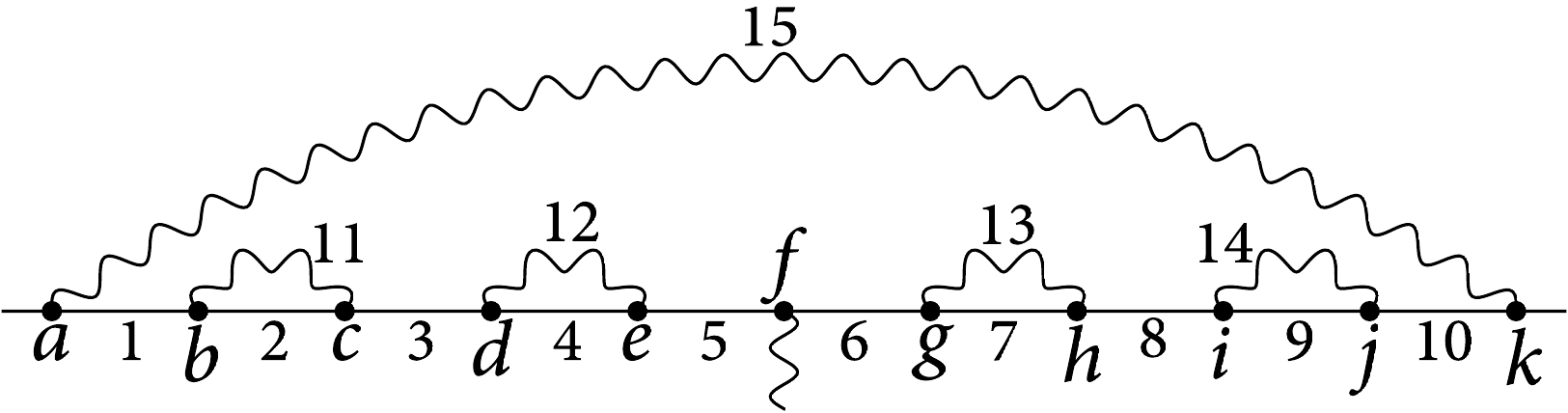}
\caption{Example with two SE-chains.} \label{fig_chains}
\end{center}
\end{figure}

\subsection{Fast sampling algorithm}

\subsubsection{Preliminaries}

Suppose the numbers $\ffdeg(s)$ are fixed for each $s\subseteq
\Lambda$, $s\neq\Lambda$, $s\neq\emptyset$, where
$\Lambda=\{1,2,\ldots,n\}$.

To generate randomly a point $(z_1,\ldots,z_n)$ it is necessary to
take two steps:
\begin{itemize}
\item generate randomly a sector $S_{j_1,\ldots,j_n}$
\item generate a point inside this sector
\end{itemize}
We generate a sector without brute forcing all sectors\footnote{In
5-loop case we have $n=14$ (see Section \ref{subsec_technical}) and
87178291200 sectors for each of 389 families of Feynman graphs.
However, what is needed is only to take $2^n=16384$ subsets for each
family.} at all stages of the calculation. We use the dynamic
programming approach instead at the initialization stage. To
generate sectors with correct probabilities it is required to know
the value
\begin{equation}\label{eq_g0_integral}
\int_{\substack{z_1,\ldots,z_n>0 \\ (z_1,\ldots,z_n)\in
S_{j_1,\ldots,j_n}}}
g_0(z_1,\ldots,z_n)\delta(z_1+\ldots+z_n-1)dz_1\ldots dz_n,
\end{equation}
where $g_0$ is defined by (\ref{eq_mc_g0}), for each sector
$S_{j_1,\ldots,j_n}$. The following lemma is used for obtaining this
integral.
\begin{lemma}\label{lemma_sum_to_max}Let $Y\subseteq \mathbb{R}^{n-1}$, $X$ be the image of
$Y$ under the map
$$
(y_2,\ldots,y_n)\rightarrow\left(\frac{y_2}{1+y_2+\ldots+y_n},\ldots,\frac{y_n}{1+y_2+\ldots+y_n}\right),
$$
$h:\mathbb{R}^n\rightarrow \mathbb{R}$ be a function satisfying
$h(k\underline{z})=h(\underline{z})/k^n$. Then
$$
\int_Y h(1,y_2,\ldots,y_n)dy_2\ldots dy_n = \int_X
h(1-x_2-\ldots-x_n,x_2,\ldots,x_n) dx_2\ldots dx_n.
$$
\end{lemma}
\begin{proof}Let us use the substitution
$x_j=y_j/(1+y_2+\ldots+y_n)$. To apply the change of variables
theorem we should prove the following relation for the Jacobian:
\begin{equation}\label{eq_jacobian}
\left| \frac{\partial (x_2,\ldots,x_n)}{\partial (y_2,\ldots,y_n)}
\right| =
\frac{h(1,y_2,\ldots,y_n)}{h(1-x_2-\ldots-x_n,x_2,\ldots,x_n)}
\end{equation}
The right part of (\ref{eq_jacobian}) equals
$1/(1+y_2+\ldots+y_n)^n$. The left part equals $|D|$, where
$$
D=\det(M'+M''),\quad M'=\parallel m'_{ij}
\parallel,\quad M''=\parallel m''_{ij}
\parallel,\quad 2\leq i,j\leq n, $$ $$
m'_{ij}=\frac{\delta_{ij}}{1+y_2+\ldots+y_n},\quad
m''_{ij}=-\frac{y_i}{(1+y_2+\ldots+y_n)^2}.
$$
The determinant $\det(M'+M'')$ is equal to the sum of the
determinants of the matrixes that are obtained from $M'$ by changing
some rows to the corresponding rows of $M''$. By $d_l$ we denote the
contribution of the matrixes that are obtained by changing $l$ rows.
It is easy to see that $d_l=0$ for $l\geq 2$, because all rows of
$M''$ are collinear. Also, it is obvious that
$$
d_0=\det M'=\frac{1}{(1+y_2+\ldots+y_n)^{n-1}}.
$$
By simple manipulations we obtain
$$
d_1=-\frac{y_2+\ldots+y_n}{(1+y_2+\ldots+y_n)^n}.
$$
Thus,
$$
D=d_0+d_1=\frac{1}{(1+y_2+\ldots+y_n)^n}.
$$
This completes the proof.
\end{proof}
Using the proved lemma and the substitution
\begin{equation}\label{eq_subst_zy}
z_{j_l}=\frac{y_l}{1+y_2+\ldots+y_n},\quad 1\leq l\leq n,
\end{equation}
where $y_1=1$, we obtain that (\ref{eq_g0_integral}) equals
$$
\int_{1\geq y_2\geq y_3\geq\ldots\geq y_n>0} \frac{\prod_{l=2}^n
(y_l/y_{l-1})^{\ffdeg(\{j_l,j_{l+1},\ldots,j_n\})}}{y_2\ldots y_n}
dy_2\ldots dy_n.
$$
By the substitution
\begin{equation}\label{eq_subst_yt}
t_l=y_l/y_{l-1}
\end{equation}
we obtain that it equals
$$
\frac{1}{\prod_{l=2}^n \ffdeg(\{j_l,j_{l+1},\ldots,j_n\})}.
$$

For generating sector permutations element-by-element we will use
the function
$$
W(s)=\sum_{\substack{j_1,\ldots,j_{|s|}\in s \\ \text{are distinct}
}} \frac{1}{\prod_{l=1}^{|s|}
\ffdeg(\{j_l,j_{l+1},\ldots,j_{|s|}\})}
$$
that is defined on all proper subsets of $\Lambda$. The function $W$
satisfies the recurrence relations:
\begin{equation}\label{eq_dyn_prog}
W(s)=\frac{\sum_{l\in s} W(s\backslash\{l\})}{\ffdeg(s)},\
s\neq\emptyset,\quad W(\emptyset)=1.
\end{equation}
When the permutation prefix $j_1,j_2,\ldots,j_{l-1}$ has already
been generated, the probability that $j_l=a$ is equal to
$P[\Lambda\backslash\{j_1,\ldots,j_{l-1}\},a]$, where
\begin{equation}\label{eq_next_prob}
P[s,a]=\frac{W(s\backslash\{a\})}{\sum_{a'\in s}
W(s\backslash\{a'\})}.
\end{equation}

Lemma \ref{lemma_sum_to_max} is also useful for generating a point
inside the given sector. By this lemma, the generation of
$z_1,\ldots,z_n$ in $S_{j_1,\ldots,j_n}$ is equivanent to the
generation of $1\geq y_2\geq y_3\geq\ldots\geq y_n$ with the
probability density
$$
C\cdot \frac{\prod_{l=2}^n
(y_l/y_{l-1})^{\ffdeg(\{j_l,j_{l+1},\ldots,j_n\})}}{y_2\ldots y_n},
$$
where the substitution (\ref{eq_subst_zy}) is applied. Applying
(\ref{eq_subst_yt}) we obtain that the generation is equivalent to
the independent generation of $t_l$, $0\leq t_l\leq 1$, $2\leq l\leq
n$ with the probability densities
$$
C\cdot t^{\ffdeg(\{j_l,\ldots,j_n\})-1}.
$$

For calculating the probability density at a given point it is
needed to know the whole integral
$$
\int_{z_1,\ldots,z_n>0}g_0(z_1,\ldots,z_n)\delta(z_1+\ldots+z_n-1)
dz_1 \ldots dz_n.
$$
It equals
$$
\sum_{a\in\Lambda} W(\Lambda \backslash \{a\}).
$$

\subsubsection{The algorithm}

\noindent\emph{Initialization part.}
\begin{enumerate}
\item Calculate $W(s)$ for all $s\subseteq\Lambda$ using (\ref{eq_dyn_prog}).
\item Calculate $P[s,a]$ for all $s\subseteq\Lambda$, $s\neq
\emptyset$, $a\in s$ using (\ref{eq_next_prob}).
\end{enumerate}

\noindent\emph{Generation part.}
\begin{enumerate}
\item Generation of a sector.\\
\verb"for "$l$\verb" := "$1$\verb" to "$n$\verb" do" \\
\verb"  put "$j_l=a$\verb" with the probability "
$P[\Lambda\backslash\{j_1,\ldots,j_{l-1}\},a]$ \verb";"
\item Generation of a point.
\begin{itemize}
\item Generate $r_2,\ldots,r_n\in [0;1]$ using the uniform
distribution.
\item Put
$t_l=r_l^{1/\ffdeg(\{j_l,\ldots,j_n\})}$, $2\leq l\leq n$.
\item Put $y_1=1$, $y_l=t_2\ldots t_l$, $2\leq l\leq n$.
\item Calculate $z_1,\ldots,z_n$ using (\ref{eq_subst_zy}).
\end{itemize}
\end{enumerate}

\subsection{Stabilization and prevention of error
underestimating}\label{subsec_stabil}

Since $f(\underline{z})/g(\underline{z})$ may be unbounded, the
integration process can crash down at any moment of time due to an
extremely big contribution of a sample. To prevent this situation we
use the following procedures.
\begin{enumerate}
\item When we generate a value $r\in [0;1]$ with uniform
distribution, we reject all $r<1/(N_{\text{gen}}+1000)^2$, where
$N_{\text{gen}}$ is the total number of generations at this moment.
No random number will be rejected during the whole process of
integration with the probability more than $99\%$ since
$$
\sum_{N=0}^{+\infty} \frac{1}{(N+1000)^2} < 0.01.
$$
However, the rejection prevents an emergence of a very small values
of $r$ that usually don't appear in batches of this size.
\item We store the variable \verb"absbound" (that is initialized by $0$), each value
$x=f(\underline{z})/g(\underline{z})$ not satisfying $|x|\leq b$,
where $b=\max(\text{\texttt{absbound}},0.1\sigma N)$, is saturated.
After
each saturation we increase \verb"absbound": \\
\verb"  absbound := "$2\cdot b$\verb";" \\
Here $\sigma$ is the current value of the standard deviation, $N$ is
the number of samples processed. This saturation prevents from
occasional appearance of extremely big values, but allows systematic
appearance of them.
\end{enumerate}

The integration error can be estimated by\footnote{For obtaining the
proper standard deviation we should also subtract
$$
\left(\sum_{j=1}^N \frac{f(\underline{z}_j)}{g(\underline{z}_j)}
\right)^2/N^3.
$$
However, in the current version of the integration program this has
not been implemented. Let us remark that in most of cases for
multiloop Feynman-parametric integrals this correction is very
small.}
$$
\sigma_{\downarrow}^2=\frac{\sum_{j=1}^N
(f(\underline{z}_j)/g(\underline{z}_j))^2}{N^2}.
$$
However, we can get an underestimation, because of:
\begin{itemize}
\item $\sigma$ can be underestimated due to a big uncertainty of
$\sigma_{\downarrow}$ connected with a small number of samples;
\item the distribution of the sample average can be quite far from
the Gaussian normal distribution.
\end{itemize}
Here we prevent only the first type of underestimation.

By definition, put
$$
n_a=|\{1\leq j\leq N:\ 2^{a-1/2}\leq
|f(\underline{z}_j)/g(\underline{z}_j)|<2^{a+1/2}\}|,\quad
a\in\mathbb{Z}.
$$
Let $a_{\text{max}}$ be the maximal $a$ such that $n_a\geq 1$. Put
$$
\triangle_{\text{uncert}}=4\cdot\max_{a=a_{\text{max}}-9}^{a_{\text{max}}}\sqrt{n_a}4^a,
$$
$$
\triangle_{\text{peak}}=\begin{cases}4^{a_{\text{max}}+d-1},\text{
if }d \geq 2, \\ 0 \text{ otherwise,} \end{cases}
$$
where $d$ is the maximal integer number such that $0\leq d\leq 6$
and for all $a$ such that $a_{\text{max}}-d<a\leq a_{\text{max}}$ we
have $n_a\leq 2$,
$$
\sigma^2_{\uparrow}=\sigma^2_{\downarrow}+\triangle_{\text{uncert}}+\triangle_{\text{peak}}.
$$
$\sigma_{\uparrow}$ is an improved estimation of $\sigma$. Here
$\triangle_{\text{uncert}}$ corresponds to the uncertainty of the
numbers $n_a$, $\triangle_{\text{peak}}$ corresponds to hypothetical
undiscovered peaks.

If $\sigma_{\uparrow}/\sigma_{\downarrow}$ is far from $1$, then
both $\sigma_{\uparrow}$ and $\sigma_{\downarrow}$ are unreliable. A
slow convergence of $\sigma_{\uparrow}/\sigma_{\downarrow}$
indicates\footnote{These indications should be considered only as
heuristics, not as rules.} that the integral $\int_{\Omega}
\frac{f(\underline{x})^2}{g(\underline{x})}d\underline{x}$ is
\textquotedblleft near to divergent\textquotedblright\footnote{For
example, the integral $\int_0^1 x^{a-1}dx$ is \textquotedblleft near
to divergent\textquotedblright\ if $a>0$ is near to zero.}, the
divergence indicates that $V(f,g)$ is infinite.

We use $\sigma_{\uparrow}$ as an estimation of $\sigma$ in all
tables of Section \ref{sec_results}.

\section{NUMERICAL RESULTS}\label{sec_results}

\subsection{Technical remarks}\label{subsec_technical}

We have evaluated the contributions of some Feynman graphs
numerically. The aim of the computation was only the test of the
method, not an obtainment of new accurate results.

For computing $A_1^{(2n)}[\text{no lepton loops}]$ we aggregate all
corresponding Feynman graphs into families. Each family corresponds
to a self-energy graph\footnote{Unlike
~\cite{kinoshita_10_new,carroll,carrollyao}, we don't work with
self-energy graphs. Self-energy graphs play only the role of
signatures of graph families. All calculations are performed with
vertex-like graphs.}. All graphs of a family are obtained from the
corresponding self-energy graph by inserting an external photon line
into an arbitrary place. The graphs, belonging to one family, have a
lot of same construction blocks in formulas and have similar numbers
$\ffdeg(s)$. Thus, the aggregation can reduce the computer time that
is needed for the calculation. We decrease the number of integration
variables by one using the idea from ~\cite{kinoshita_rules}: for
each graph $G$ we treat the sum $z_a+z_b$ as one variable, where
$a,b$ are the electron lines that are incident to the vertex that is
incident to the external photon line\footnote{It was observed that
the integrands from Section \ref{sec_integrands} depend linearly on
$z_a$ when $z_a+z_b$ is fixed.}. This allows us to use a unified set
of integration variables for each family of graphs (each integration
variable corresponds to an internal line of the self-energy graph of
the family). For a family $M$ we use the values
$$
\ffdeg(s)=\min_{G\in M} \ffdeg_G(\{j:\ l_G(j)\in s\}),
$$
where $l_G(j)$ is the line in $G_M$ that corresponds to the line $j$
in $G$, where $G_M$ is the self-energy graph corresponding to $M$;
here by $\ffdeg_G(s)$ we denote the value $\ffdeg(s)$ that is
constructed in the graph $G$.

The values (\ref{eq_mc_constants}) were used for Monte Carlo in all
calculations. The described importance sampling method was combined
with the adaptive algorithm: the whole integration area was split
into subsets, each subset contains all sectors $S_{j_1,\ldots,j_n}$
with the fixed $(j_1,j_2)$; for each subset the probability of
selecting this subset is adjusted dynamically during the integration
to minimize $\sigma_{\uparrow}$. The value $\sigma_{\uparrow}$ was
first calculated for each subset separately, and the values were
combined after this. Before dynamical adjusting, each subset is
initialized by 50 Monte Carlo samples. The splitting improved
$\sigma$ by about $1.3\ldots 1.5$ times.

The D programming language ~\cite{dlang} was used for the generator
of the code of the integrands and for the Monte Carlo integrator.
The code of the integrands was generated in the C++ programming
language. Total size of the C++ generated code for the 4-loop
integrands is 230 MB. The corresponding size of the compiled code is
600 MB. Interval arithmetic was used for preventing round-off
errors\footnote{For more detailed explanation about the nature of
these round-off errors, see ~\cite{volkov_2015}.}:
\begin{itemize}
\item each value is represented as an interval; it is supposed that
the exact value is in the interval;
\item arithmetic operations on intervals are defined in such a way
as to preserve this property.
\end{itemize}
The value of an integrand at a point was first calculated as an
interval in the machine 64-bit precision. If the precision was not
enough, it was evaluated as an interval with the 352-bit precision.
The points, for which the 352-bit precision is not enough, are
ignored. Machine-precision and arbitrary-precision interval
arithmetic calculations were performed with the help of Branimir
Lambov's RealLib.

Two computer configurations were used for the computations. The
configuration A is 1 core of AMD Athlon(tm) II P320 2.1GHz. The
configuration B is 2 cores of Intel Xeon E5-2658A, 2.2GHz.

\subsection{Results of computations}\label{subsec_results}

Table \ref{table_total} contains the numerical results of computing
$A_1^{(2n)}[\text{no lepton loops}]$. The comparison with the known
analytical values\footnote{For $n=4$ we compare with the recent
result of S.Laporta ~\cite{laporta_8}. This result is in a good
agreement with the results from
~\cite{kinoshita_10_new,kinoshita_8_revis}.} is provided. Here,
\begin{itemize}
\item $N_l$ is the number of independent loops;
\item Val. is the computed value with the estimated error ($1\sigma$
limits);
\item An.val. is the known analytical or semi-analytical value;
\item Ref. is the references to the papers where the analytical value
is presented;
\item $N_{\text{call}}$ is the total number of calls of the integrand functions
(i.e., this is the number of Monte Carlo samples); \item
$N_{\text{prec}}$ is the number of Monte Carlo samples for which the
machine 64-bit precision was not enough (see Section
\ref{subsec_technical});
\item $\sigma_{\uparrow}/\sigma_{\downarrow}$ is the relation
between the corrected and the direct estimations of $\sigma$ (see
Section \ref{subsec_stabil});
\item comp. is the computer configuration (A or B, see Section
\ref{subsec_technical}) and the time of the computation (h=hours,
d=days).
\end{itemize}

Table \ref{table_4loop} contains the contributions of the individual
families of Feynman graphs to $A_1^{(8)}[\text{no lepton loops}]$.
The self-energy graphs corresponding to the families are shown in
FIG. \ref{fig_4loop}.

We also have evaluated the individual contributions of the ladder
graphs (FIG. \ref{fig_ladder}) and the fully crossed ladder graphs
(FIG. \ref{fig_crosses}) up to 5 loops for testing the method and
for comparing with the known values. The size of C++ generated code
is 2.4 MB for the 5-loop ladder graph and 10 MB for the 5-loop fully
crossed ladder graph. The corresponding sizes of the compiled code
are 6.2 MB and 24 MB. The contributions of the ladder graphs that
are obtained by the presented method are the same as the
contributions that are obtained by the standard subtractive on-shell
renormalization\footnote{However, the standard renormalization
doesn't lead to finite Feynman-parametric integrals, see
~\cite{volkov_2015}.}. This fact can be proved by simple algebraic
transformations, see the Section 3 of ~\cite{volkov_2015}. The fully
crossed ladder graphs don't contain divergent subgraphs. Thus, the
contributions of the fully crossed ladder graphs don't depend on the
kind of a subtraction procedure. These graphs are a direct test for
Monte Carlo integration. The results for the ladder graphs and for
the fully crossed ladder graphs are provided in Table
\ref{table_ladder} and Table \ref{table_cross} respectively. Here,
$\triangle_{\text{prec}}$ is the contribution of the points for
which the machine 64-bit precision was not enough. The results for
4-loop and 5-loop fully crossed ladder graphs are new. The
dependence of the precision of that 5-loop calculations on number of
Monte Carlo samples\footnote{Initialization samples are included in
$N_{\text{call}}$.} is shown in Tables \ref{table_ladder_nsamples},
\ref{table_cross_nsamples} (by Diff. we mean the difference between
the obtained value and the analytical one from ~\cite{caffo}).

\small \ctable[pos=H,label=table_total, caption={Numerical results
($1\sigma$ limits) for $A_1^{(2n)}{[\text{no lepton loops}]}$ and
comparison with known analytical values.}]{cccccccc}{}{\hline \hline
$N_l$ & Val. & An.val. & Ref. & $N_{\text{call}}$ &
$N_{\text{prec}}$ & $\sigma_{\uparrow}/\sigma_{\downarrow}$ & comp.
\\
\hline 2 & $-0.34416(10)$ & $-0.344167$ &
~\cite{analyt2_p,analyt2_z2} & $95\cdot 10^7$ & $10116$ & $1.002$ &
A,6.5h \\
3 & $0.9019(55)$ & $0.90437$ &
~\cite{analyt_b,analyt_e,analyt_d,analyt_c,analyt_f,analyt3} &
$43\cdot 10^7$ & $60466$ & $1.044$ & A,24h \\
4 & $-2.34(17)$ & $-2.1769$ & ~\cite{laporta_8} & $10^9$ & $68\cdot
10^4$ & $1.17$ & B,16d \\ \hline \hline } \normalsize

\ctable[pos=H,label=table_4loop,caption={Contributions ($1\sigma$
limits) of the families from FIG. \ref{fig_4loop} to
$A_1^{(8)}{[\text{no lepton loops}]}$.}]{cccccccc}{}{ \hline \hline
\# & value & $N_{\text{call}}$ & $\sigma_{\uparrow}/\sigma_{\downarrow}$& \# & value & $N_{\text{call}}$ & $\sigma_{\uparrow}/\sigma_{\downarrow}$\\
\hline $M_{01}$ & $-1.198(30)$ & $30\cdot 10^{6}$ & 1.20
 & $M_{25}$ & $-1.266(14)$ & $87\cdot 10^{5}$ & 1.30
\\
$M_{02}$ & $-1.689(40)$ & $47\cdot 10^{6}$ & 1.19
 & $M_{26}$ & $-1.598(34)$ & $36\cdot 10^{6}$ & 1.18
\\
$M_{03}$ & $-2.537(35)$ & $41\cdot 10^{6}$ & 1.15
 & $M_{27}$ & $-0.246(20)$ & $16\cdot 10^{6}$ & 1.17
\\
$M_{04}$ & $6.473(48)$ & $71\cdot 10^{6}$ & 1.14
 & $M_{28}$ & $6.049(43)$ & $58\cdot 10^{6}$ & 1.13
\\
$M_{05}$ & $4.202(17)$ & $12\cdot 10^{6}$ & 1.18
 & $M_{29}$ & $1.585(19)$ & $14\cdot 10^{6}$ & 1.20
\\
$M_{06}$ & $-0.990(26)$ & $24\cdot 10^{6}$ & 1.18
 & $M_{30}$ & $-3.057(33)$ & $37\cdot 10^{6}$ & 1.14
\\
$M_{07}$ & $-2.170(26)$ & $23\cdot 10^{6}$ & 1.25
 & $M_{31}$ & $0.8248(76)$ & $54\cdot 10^{5}$ & 1.20
\\
$M_{08}$ & $-5.282(42)$ & $56\cdot 10^{6}$ & 1.16
 & $M_{32}$ & $-0.6672(76)$ & $52\cdot 10^{5}$ & 1.22
\\
$M_{09}$ & $-1.112(27)$ & $26\cdot 10^{6}$ & 1.15
 & $M_{33}$ & $-0.7301(49)$ & $39\cdot 10^{5}$ & 1.20
\\
$M_{10}$ & $1.845(37)$ & $41\cdot 10^{6}$ & 1.15
 & $M_{34}$ & $1.084(11)$ & $74\cdot 10^{5}$ & 1.19
\\
$M_{11}$ & $3.244(31)$ & $31\cdot 10^{6}$ & 1.18
 & $M_{35}$ & $1.859(11)$ & $70\cdot 10^{5}$ & 1.20
\\
$M_{12}$ & $-3.633(36)$ & $40\cdot 10^{6}$ & 1.14
 & $M_{36}$ & $-1.619(15)$ & $10^{7}$ & 1.30
\\
$M_{13}$ & $-4.337(17)$ & $13\cdot 10^{6}$ & 1.15
 & $M_{37}$ & $0.9384(54)$ & $42\cdot 10^{5}$ & 1.22
\\
$M_{14}$ & $0.580(29)$ & $26\cdot 10^{6}$ & 1.23
 & $M_{38}$ & $-4.209(14)$ & $93\cdot 10^{5}$ & 1.23
\\
$M_{15}$ & $0.308(28)$ & $27\cdot 10^{6}$ & 1.21
 & $M_{39}$ & $-1.982(14)$ & $98\cdot 10^{5}$ & 1.32
\\
$M_{16}$ & $7.507(44)$ & $64\cdot 10^{6}$ & 1.14
 & $M_{40}$ & $1.548(24)$ & $18\cdot 10^{6}$ & 1.29
\\
$M_{17}$ & $2.895(27)$ & $26\cdot 10^{6}$ & 1.12
 & $M_{41}$ & $-1.892(19)$ & $15\cdot 10^{6}$ & 1.20
\\
$M_{18}$ & $-4.827(36)$ & $42\cdot 10^{6}$ & 1.14
 & $M_{42}$ & $2.262(23)$ & $20\cdot 10^{6}$ & 1.19
\\
$M_{19}$ & $0.4035(65)$ & $47\cdot 10^{5}$ & 1.19
 & $M_{43}$ & $-1.1308(89)$ & $57\cdot 10^{5}$ & 1.21
\\
$M_{20}$ & $2.219(13)$ & $86\cdot 10^{5}$ & 1.18
 & $M_{44}$ & $2.312(17)$ & $13\cdot 10^{6}$ & 1.22
\\
$M_{21}$ & $0.6548(55)$ & $41\cdot 10^{5}$ & 1.22
 & $M_{45}$ & $2.109(20)$ & $16\cdot 10^{6}$ & 1.19
\\
$M_{22}$ & $-1.721(16)$ & $11\cdot 10^{6}$ & 1.21
 & $M_{46}$ & $-0.049(15)$ & $10^{7}$ & 1.23
\\
$M_{23}$ & $-3.150(15)$ & $10^{7}$ & 1.19
 & $M_{47}$ & $-2.113(21)$ & $17\cdot 10^{6}$ & 1.19
\\
$M_{24}$ & $-0.035(19)$ & $15\cdot 10^{6}$ & 1.21
 & \  & \  & \  & \
\\
\hline \hline }

\small \ctable[pos=H,label=table_ladder,caption={Numerical results
($1\sigma$ limits) for the ladder Feynman graphs up to 5 loops and
comparison with known analytical values.}]{ccccccccc}{}{\hline
\hline $N_{\text{l}}$ & Val. & An.val. & Ref. & $N_{\text{call}}$ &
$N_{\text{prec}}$ & $\triangle_{\text{prec}}$ &
$\sigma_{\uparrow}/\sigma_{\downarrow}$ & comp.
\\
\hline 2 & $0.777440(67)$ & $0.777478$ & ~\cite{analyt2_p,caffo} &
$62\cdot 10^{7}$ & $3160$ & $0.00005$ & $1.001$ & A,3h
\\
3 & $1.79052(37)$ & $1.790278$ & ~\cite{analyt_d,caffo} & $68\cdot
10^{7}$ & $15664$ & $0.0011$ & $1.01$ & A,11h
\\
4 & $4.3035(39)$ & $4.29765$ & ~\cite{caffo} & $39\cdot 10^{7}$ &
$18997$ & $0.018$ & $1.08$ & A,24h
\\
5 & $11.681(46)$ & $11.6592$ & ~\cite{caffo} & $4\cdot 10^{8}$ &
$30652$ & $0.25$ & $1.32$ & A,7.5d
\\
\hline \hline } \normalsize

\small \ctable[pos=H,label=table_cross,caption={Numerical results
($1\sigma$ limits) for the fully crossed ladder Feynman graphs up to
5 loops and comparison with known analytical
values.}]{ccccccccc}{}{\hline \hline $N_{\text{l}}$ & Val. & An.val.
& Ref. & $N_{\text{call}}$ & $N_{\text{prec}}$ &
$\triangle_{\text{prec}}$ & $\sigma_{\uparrow}/\sigma_{\downarrow}$
& comp.
\\
\hline 2 & $-0.467666(49)$ & $-0.467645$ & ~\cite{analyt2_p} &
$61\cdot 10^{7}$ & $169$ & $-4\cdot 10^{-7}$ & $1.003$ & A,2.5h
\\
3 & $-0.026810(47)$ & $-0.026800$ & ~\cite{analyt3} & $56\cdot
10^{7}$ & $16370$ & $-5\cdot 10^{-5}$ & $1.008$ & A,10h
\\
4 & $0.29685(21)$ & - & - & $14\cdot 10^{7}$ & $25009$ & $0.0014$ &
$1.057$ & A,26h
\\
5 & $-0.6427(21)$ & - & - & $25\cdot 10^{6}$ & $14033$ & $-0.018$ &
$1.228$ & A,2.5d
\\
\hline \hline} \normalsize

\small \ctable[pos=H,label=table_ladder_nsamples,caption={Dependence
of the 5-loop ladder graph precision on
$N_{\text{call}}$.}]{cccccc}{}{\hline \hline $N_{\text{call}}$ &
Val. & $\sigma_{\uparrow}$ & $\sigma_{\downarrow}$ & Diff. &
$\sigma_{\uparrow}/\sigma_{\downarrow}$
\\
\hline $10^4$ & $5.02$ & $37.4$ & $1.85$ & $-6.64$ & $20.2$
\\
$4\cdot 10^4$ & $5.79$ & $2.21$ & $1.2$ & $-5.87$ & $1.84$
\\
$25\cdot 10^4$ & $9.2$ & $1.28$ & $0.7$ & $-2.45$ & $1.82$
\\
$10^6$ & $10.42$ & $0.64$ & $0.42$ & $-1.24$ & $1.51$
\\
$4\cdot 10^6$ & $10.67$ & $0.33$ & $0.23$ & $-0.98$ & $1.45$
\\
$25\cdot 10^6$ & $11.236$ & $0.167$ & $0.112$ & $-0.424$ & $1.49$
\\
$10^8$ & $11.621$ & $0.089$ & $0.064$ & $-0.038$ & $1.38$
\\
$4\cdot 10^8$ & $11.6816$ & $0.046$ & $0.0349$ & $0.0224$ & $1.32$
\\
\hline \hline } \normalsize

\small \ctable[pos=H,label=table_cross_nsamples,caption={Dependence
of the 5-loop fully crossed ladder graph $\sigma$ on
$N_{\text{call}}$.}]{ccccc}{}{\hline \hline $N_{\text{call}}$ & Val.
& $\sigma_{\uparrow}$ & $\sigma_{\downarrow}$ &
$\sigma_{\uparrow}/\sigma_{\downarrow}$
\\
\hline $10^4$ & $-0.664$ & $0.267$ & $0.083$ & $3.2$
\\
$4\cdot 10^4$ & $-0.619$ & $0.054$ & $0.027$ & $2$
\\
$25\cdot 10^4$ & $-0.6146$ & $0.02$ & $0.0124$ & $1.609$
\\
$10^6$ & $-0.6318$ & $0.0105$ & $0.0072$ & $1.462$
\\
$4\cdot 10^6$ & $-0.6369$ & $0.00539$ & $0.00397$ & $1.358$
\\
$25\cdot 10^6$ & $-0.6427$ & $0.00213$ & $0.00173$ & $1.228$
\\
\hline \hline } \normalsize

\begin{figure}[H]
\includegraphics{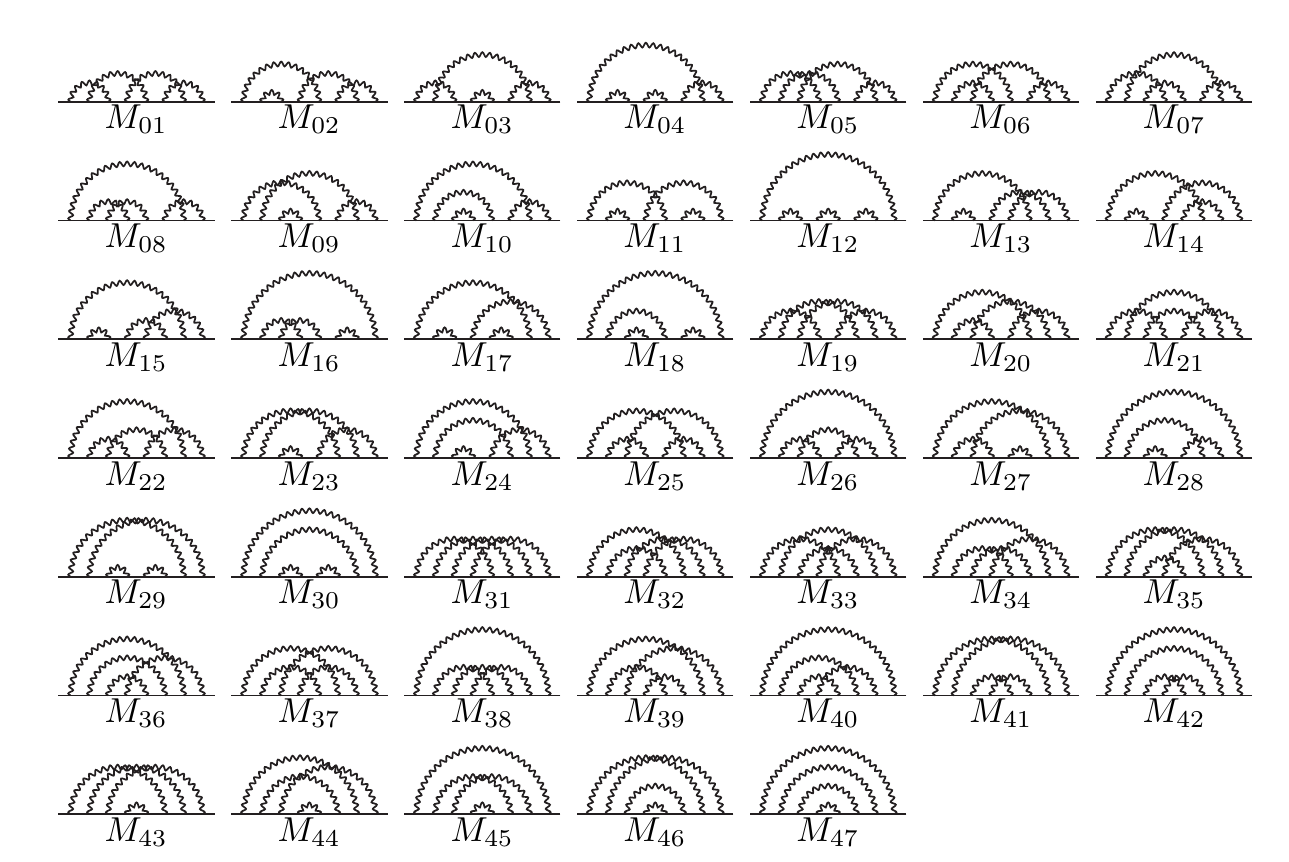}
\caption{Families of 4-loop Feynman graphs without lepton loops for
AMM, the numeration is taken from ~\cite{kinoshita_8_revis}.}
\label{fig_4loop}
\end{figure}

\begin{figure}[H]
\includegraphics{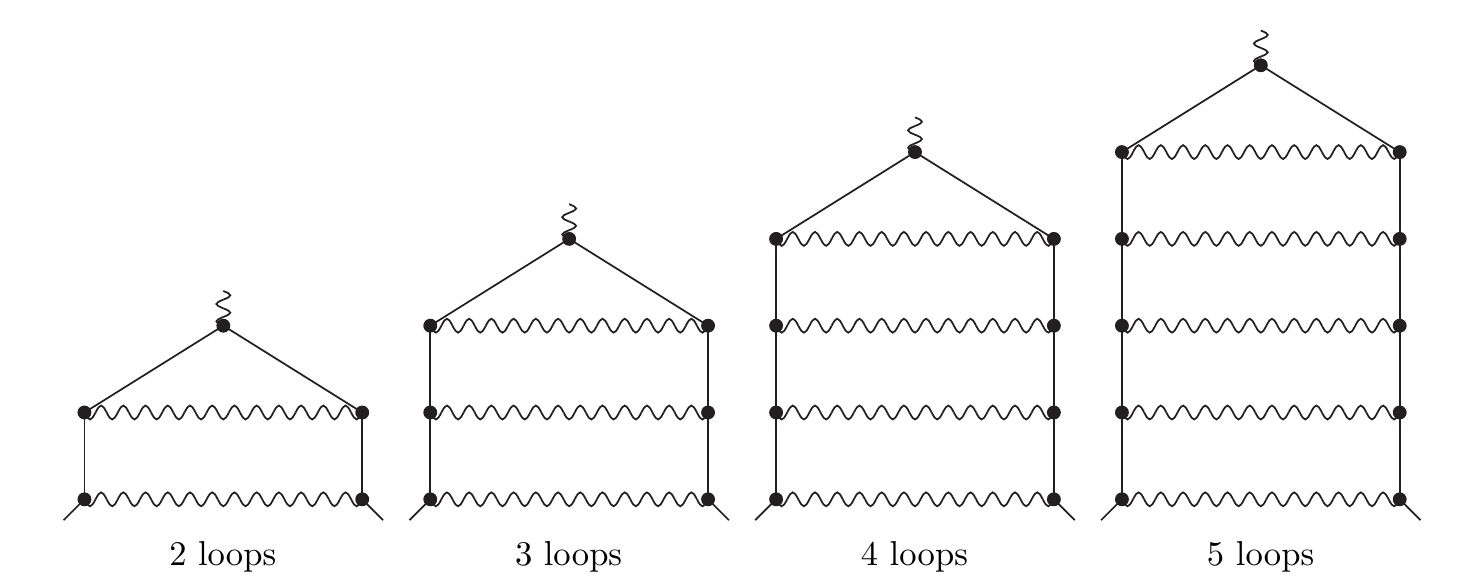}
\caption{Ladder Feynman graphs.} \label{fig_ladder}
\end{figure}

\begin{figure}[H]
\includegraphics{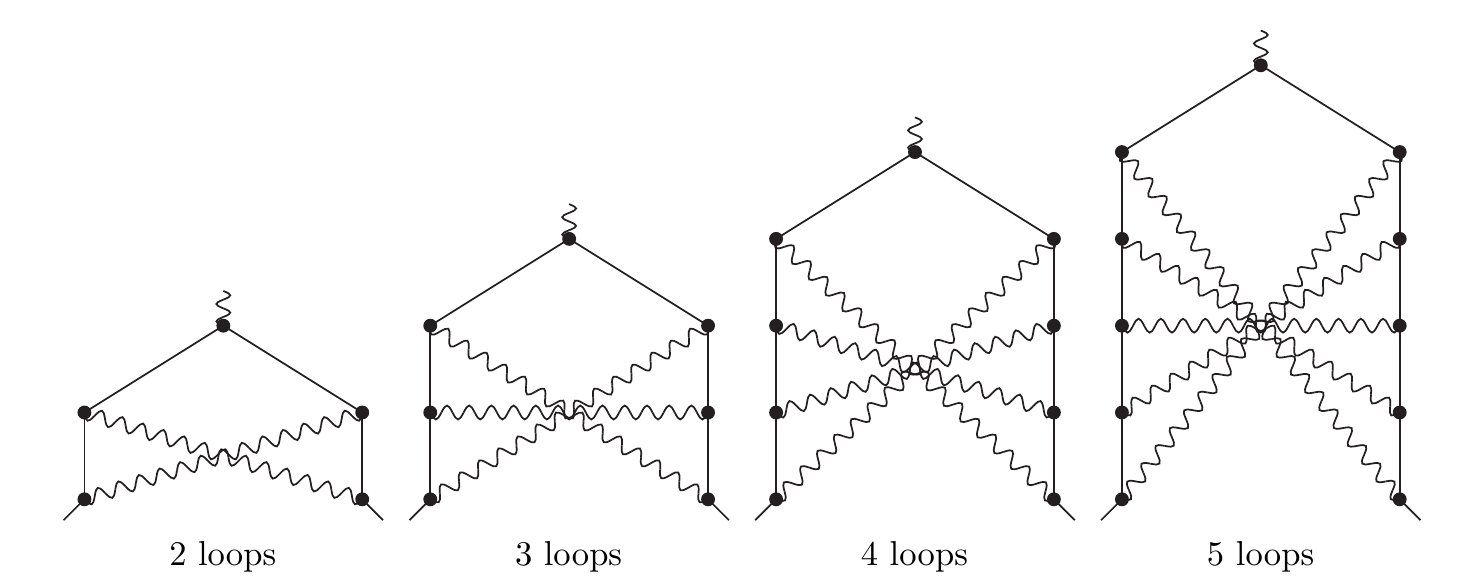}
\caption{Fully crossed ladder Feynman graphs.} \label{fig_crosses}
\end{figure}

\subsection{Comparison with other methods with respect to Monte
Carlo convergence speed}\label{subsec_speed}

Table \ref{table_mc} contains the comparison of the presented method
with other ones with respect to Monte Carlo convergence speed. We
suppose that $\sigma \sim C/\sqrt{N}$, where $N$ is the number of
Monte Carlo samples. Using the value of $C$ we can estimate the
convergence speed. The table shows that this method has an advantage
over the others. However, we must take into account the following.
\begin{itemize}
\item The calculations ~\cite{carrollyao,carroll},
~\cite{kinoshita_6_prec}, ~\cite{kinoshita_8_wrong} are not recent.
Recent calculations can have improvements in Monte Carlo integration
(however, the information about the number of Monte Carlo samples
for $A_1^{(8)}[\text{no lepton loops}]$ was not provided).
\item The information in ~\cite{carrollyao,carroll},
~\cite{kinoshita_6_prec}, ~\cite{kinoshita_8_wrong} about the number
of samples is very rough.
\item The integrands in this calculation and in ~\cite{carrollyao,carroll},
~\cite{kinoshita_6_prec}, ~\cite{kinoshita_8_wrong} have a different
nature due to the difference in subtraction procedures and in the
ways of extracting AMM. It is possible, theoretically, that a slow
convergence with respect to number of samples can be compensated by
a fast evaluation of integrands.
\item In table \ref{table_mc} we don't take into account the number
of samples that were evaluated for residual renormalization in
~\cite{carrollyao,carroll}, ~\cite{kinoshita_6_prec},
~\cite{kinoshita_8_wrong}. The presented method doesn't need
residual renormalizations.
\item We can have an error underestimation due to a small number of
samples. Also, $\sigma\cdot\sqrt{N_{\text{call}}}$ can increase with
the rise of $N_{\text{call}}$.
\item The error in ~\cite{carrollyao,carroll} was underestimated by
about $2.7$ times. The reason is unknown.
\item The calculation ~\cite{kinoshita_8_wrong} was wrong due to an
algebraic error ~\cite{kinoshita_8_revis}.
\end{itemize}

\small \ctable[pos=H,label=table_mc,caption={Comparison of this
method of calculating $A_1^{(2n)}{[\text{no lepton loops}]}$ with
the others with respect to Monte Carlo convergence
speed.}]{ccccc}{}{ \hline \hline Calculation & Val. & $\sigma$ &
$N_{\text{call}}$ & $\sigma\cdot\sqrt{N_{\text{call}}}$ \\ \hline
$n=3$, this calculation, 8 integrands & $0.9019$ & $0.0055$ &
$43\cdot 10^7$ & $119.2$ \\ $n=3$, ~\cite{carrollyao,carroll},
8 integrands, RIWIAD & $0.74$ & $0.06$ & $16\cdot 10^6$ & $240$ \\
 $n=3$, ~\cite{kinoshita_6_prec}, 8 integrands, VEGAS & $0.904882$
& $0.000347$ & $3\cdot 10^{12}$ & $601$ \\ $n=4$, this calculation,
47 integrands & $-2.34$ & $0.17$ & $10^9$ & $5375.9$
\\ $n=4$, ~\cite{kinoshita_8_wrong}, 47 integrands, VEGAS &
$-1.99306$ & $0.00343$ & $87\cdot 10^{12}$ & $31992.9$ \\ \hline
\hline } \normalsize

\section{CONCLUSION}

The method for numerical evaluation of $A^{(2n)}_1[\text{no lepton
loops}]$ was developed. The method is based on the subtraction
procedure from ~\cite{volkov_2015} and on the new importance
sampling Monte Carlo algorithm. The method has been checked
numerically for $n=2,3,4$ on personal computers, the results are in
good agreement with the known ones. Also, the contributions of some
individual 5-loop graphs were computed. The contributions of the
ladder graphs are in good agreement with the known analytical ones
($1.65\sigma$ limits). The obtained contributions of the 4-loop and
5-loop fully crossed ladder graphs are new and can be compared in
the future. The method was compared with the other ones with respect
to Monte Carlo convergence speed. The new method gives $\sigma$
about 4 times less for $n=3$ and about 6 times less for $n=4$ when
the number of samples is fixed. This comparison is not quite correct
due to different reasons. However, this shows that the method can be
used for precise evaluation of $A^{(2n)}_1[\text{no lepton loops}]$
for $n=4$ with the help of supercomputers. The question about
effectiveness of the method for $n=5$ is still open. Also, the
following problems remain open:
\begin{itemize}
\item to prove mathematically (or disprove) that the developed subtraction
procedure leads to a finite Feynman-parametric integral for all
Feynman graphs for any $n$;
\item to prove mathematically that the given probability density
function leads to a finite variance $V(f,g)$;
\item to develop a method of obtaining $\ffdeg(s)$ for Feynman
graphs containing lepton loops\footnote{It seems that I-closures are
not suitable for graphs with lepton loops.}.
\end{itemize}

\section*{ACKNOWLEDGEMENTS}

The author thanks A. L. Kataev for fruitful discussion and helpful
recommendations, as well as A. B. Arbuzov for his help in
organizational issues and in preparation of this text.

\end{document}